\documentclass{fundam}

\usepackage[hidelinks]{hyperref}
\usepackage{amsmath}
\usepackage{amssymb}
\usepackage{amsfonts}
\usepackage{algorithm}
\usepackage{algpseudocode}
\usepackage{mathtools}
\usepackage{graphicx}
\usepackage{arydshln}

\DeclareMathOperator{\dist}{dist}
\DeclareMathOperator{\out}{out}
\DeclareMathOperator{\inedge}{in}

\DeclareMathOperator{\MWST}{MWST}

\newtheorem{problem}{Problem}

\title{Exploring Temporal Graphs with Frequent and Regular Edges}

\address{School of Computer Science, University of St Andrews, United Kingdom}

\author{
  Duncan Adamson\\
  School of Computer Science\\University of St Andrews\\KY16 9SX\\United Kingdom\\
  duncan.adamson{@}st-andrews.ac.uk
}

\runninghead{D. Adamson}{Exploring Temporal Graphs with Frequent and Regular Edges}

\begin{document}

\maketitle

\begin{abstract}
Temporal graphs are a class of graphs defined by a constant set of vertices and a changing set of edges, each of which is known as a timestep. These graphs are well motivated in modelling real-world networks, where connections may change over time. One such example, itself the primary motivation for this paper, are public transport networks, where vertices represent stops and edges the connections available at some given time. Exploration problems are one of the most studied problems for temporal graphs, asking if an agent starting at some given vertex $v$ can visit every vertex in the graph.

In this paper, we study two primary classes of temporal graphs. First, we study temporal graphs with \emph{frequent edges}, temporal graphs where each edge $e$ is active at least once every $f_e$ timesteps, called the frequency of the edge.  Second, temporal graphs with \emph{regular edges}, graphs where each edge $e$ is active at any timestep $t$ where $t \equiv s_e \bmod r_e$, with $s_e$ being the start time of the edge, and $r_e$ the regularity.

We show that graphs with frequent edges can be explored in $O(F n)$ timesteps, where $F = \max_{e \in E} f_e$, and that graphs with regular edges can be explored in $O(R n)$ timesteps, where $R = \max_{e \in E} r_e$. We provide additional results for \emph{public transport graphs}, temporal graphs formed by the union of several routes, corresponding to the schedules of some modes of transit, for \emph{sequential connection graphs}, temporal graphs in which each vertex has a single active in-edge per timestep, iterating over the set of edges in some order, and for \emph{broadcast networks}, a representation of communication within distributed networks where each vertex broadcasts a message either to all vertices, or none at each timestep.
\looseness=-1

\end{abstract}


\newpage




\section{Introduction}

In many real world settings, networks are not static objects but instead have unstable connections that vary with time. Such examples include public transport networks \cite{kutner2025better} and infection control \cite{ruget2021multi}. Temporal graphs provide a model for such time-varying networks. Formally, a \emph{temporal graph $\mathcal{G}$} is a generalisation of (static) graphs containing a common vertex set $V$ and ordered sequence of $T$ edge sets $(E_1, E_2, \dots, E_T)$, each called a \emph{timestep} (also known as a \emph{snapshot}), with the \emph{lifetime} of the graph defined by the number of edge sets. Unlike some models of dynamic graphs, we assume that we have full knowledge of the graph, i.e. that we are given every edge set as part of the input.

Much work on temporal graphs has focused on the problem of \emph{exploration} of temporal graphs. In this problem, some agent (or set of agents) must visit every vertex in a temporal graph. The agent is restricted in being allowed only a single move per timestep, thus as the graph changes, the agent movement may become restricted. This was motivated by extending the well known \textsc{Travelling Salesman} to temporal graphs by Michail and Spirakis \cite{michail2016traveling}. Since then, there has been a broad span of work covering reachability \cite{kutner2025better,DeligkasP22,EMMZ21,meeks2022reducing} and exploration \cite{michail2016traveling,adamson2022faster,arrighi2023kernelizing,dogeas2023exploiting,ERLEBACH2021,erlebach2019two,ErlebachKDefficent,erlebach2023parameterised}.

The decision version of \textsc{TEXP} in which one has to decide if at least one exploration schedule exists in a given temporal graph from a given starting vertex is an \textbf{NP}-complete problem \cite{michail2016traveling}. Indeed, this problem remains \textbf{NP}-complete even if the underlying graph has pathwidth 2 and every snapshot is a tree \cite{bodlaender2019exploring}, or if the underlying graph is a star and the exploration has to start and end at the center of the star \cite{akrida2021temporal}.

In the optimisation version of this problem, we ask for the \emph{quickest} exploration of a given temporal graph, that being the schedule ending at the earliest timestep, formally defied in Section \ref{sec:prelims}.

On the positive side, we mention the landmark paper by Erlebach et. al. \cite{ERLEBACH2021}, who show that any \emph{always connected} temporal graphs (temporal graphs where each timestep is connected) can be explored in $O(n^2)$ timesteps, where $n$ is the number of vertices. For graphs where the \emph{underlying graph} (the static graph corresponding to the union over all edge sets) is planar,  they reduce this to $O(n^{1.8} \log n)$ timesteps, strengthened by Adamson et al. \cite{adamson2022faster} to $O(n^{1.75} \log n)$ timesteps. If the underlying graph has treewidth at most $k$ then the temporal graph can be explored in $O(k^{1.5} n^{1.5} \log n)$ timesteps, later strengthened to $O(k n^{1.5} \log n)$ timesteps \cite{adamson2022faster}. If the underlying graph is a $2 \times n$ grid then the temporal graph can be explored in $O(n \log^3 n)$ timesteps. Finally that if the underlying graph is a cycle or a cycle with a single chord, then the temporal graph can be explored in $O(n)$ timesteps, later generalised by Alamouti \cite{taghian2020exploring} to exploring temporal graphs with an underlying graph as a cycle with $k$-chords, giving an $O(k^2 k! e^k n)$ upper bound on exploration, strengthed by Adamson et al. \cite{adamson2022faster} to an upper bound of $O(kn)$.
%
%
%
%
%
Erlebach and Spooner \cite{ErlebachKDefficent} showed that \emph{$k$-edge-deficient} temporal graphs, temporal graphs in which each timestep has at most $k$-edges removed from the underlying graph, can be explored in $O(k n \log n)$ timesteps, or $O(n)$ when only a single edge is removed.

On the negative side, we have, by Erlebach et al. \cite{ERLEBACH2021}, that there exists a class of always connected temporal graphs requiring $\Omega(n^2)$ timesteps to explore, and that there exists temporal graphs with an underlying planar graph of degree at most $4$ that cannot be explored faster than in $\Omega(n \log n)$ timesteps. Further, in \cite{ErlebachKDefficent} Erlebach and Spooner showed that there exist some $k$-edge-deficient temporal graphs that cannot be explored faster than in $\Omega(n \log k)$ timesteps.

In this paper, we focus on the exploration of temporal graphs with \emph{frequent} and \emph{regular} edges. We define such graphs by how often each edge appears and, more importantly in this work, the maximum number of consecutive timesteps such that a given edge is not active. More formally, given a edge, $e$, we define the frequency of $e$, $f_e$, as the smallest value such that $\forall t \in [1, T - f_e]$, $e \in \bigcup_{t' \in [t, t + f_e]} E_{t'}$, with $T$ being the lifetime of the graph. The regularity of an edge, $e$, is defined as the smallest value, $r_e$ such that $\forall t \in [1, T - r_e]$, $e \in E_{t}$ iff $e \in E_{t + r_e}$. We are motivated primarily by public transport networks, where an edge represents a connection between stations. In such systems we assume, or hope, see \cite{kutner2025better}, that the connections between stations in the system have some known frequency or preferably regularity.

The notions of frequency and regularity, by our definitions, have been studied previously with a particular focus on graph properties \cite{arrighi2023multi,erlebach2024parameterized,mertzios2024realizing,mertzios2025temporal}, Cops and Robber games\cite{de2023cop,erlebach2024cop}. Of particular relevance to us is the paper by Bellitto et al. \cite{bellitto_et_al:LIPIcs.SAND.2023.13} on restless exploration, a restriction of exploration where the agent must move every timestep if possible, of periodic temporal graphs. The authors show that determining if the graph has a restless exploration is solvable in polynomial time if the graph has a period of 2, equivalent to a regularity of 2 in our model, and NP-Hard if the period is greater than 2.

\paragraph*{Our Results}

In Section \ref{sec:main_results}, we show that graphs where each edge has a frequency of at most $f$ can be explored in $f(2n - 3)$ timesteps, and that graphs where each edge has a regularity of at most $r$ can be explored in $r(2n - 3)$ timesteps. We note that, as our definition of regularity generalises the definition of periodicity for temporal graphs given in \cite{bellitto_et_al:LIPIcs.SAND.2023.13} and of edge periodicity given in \cite{erlebach2024cop}, thus can be directly applied to these settings.

Further, in Section \ref{sec:motivating_results}, we explicitly introduce and study the following sets of temporal graphs, applying our main results to each in order to get an upper bound on the exploration time:
\begin{itemize}
    \item \emph{Public Transportation Graphs}: A class of temporal graphs defined by a given set of (temporal) walks, with the edges active at timestep $t$ corresponding exactly to the edges in the walks used in that timestep. We show that every edge in these graphs has a period of at most the length of the longest walk, $L$, and thus can be explored in $L(2n - 3)$ timesteps.
    \item \emph{Sequential Connection Graphs}: A class of temporal graphs where each vertex activates exactly one in-edge at each time step, following a given permutation of the in-edges. We show that these graphs can be explored in at most $4 \vert E \vert$ timesteps, where $E$ is the set of edges in the underlying graph.
    \item \emph{Broadcast Networks}: A class of temporal graphs where the set of out-edges from each vertex are either all active, or all inactive at a given timestep, with the additional constraint that all in-edges to a given vertex must be active between each activation of the out-edges. We show that such graphs can be explored in $O(d n^2)$ timesteps, where $d$ is diameter of the graph and, further, in $O(\delta n)$ timesteps if the graph is always connected, where $\delta$ is the minimum degree of any vertex in the graph.
\end{itemize}

\section{Preliminaries}
\label{sec:prelims}

We first define the notation used in this paper. Let $[i, j] = i, i + 1, \dots, j$ denote the (ordered) set of integers between some pair $i, j \in \mathbb{N}$, where $\mathbb{N}$ is the set of natural numbers. Note that if $i = j$, then $[i, j] = \{i \}$, and if $i > j$, $[i, j] = \emptyset$.

We define a \emph{graph} $G = (V, E)$ by a set of vertices, by convention $V = (v_1, v_2, \dots, v_n)$, and set of edges, by convention $E \subseteq V \times V$, each a tuple of vertices. Note that we may write the edge $e$ between the vertices $v_i$ and $v_j$ as either $e = (v_i, v_j)$ or $e = (v_j, v_i)$. When an edge is given explicitly as $(v_i, v_j)$, we call $v_i$ the \emph{start point} and $v_j$ the \emph{end point}.
A \emph{walk} in a graph is an ordered sequence of edges $W = (v_{i_1}, v_{i_2}) (v_{i_2}, v_{i_3}) \dots (v_{i_{m - 1}, v_{i_m}})$, where the end point the the $j^{th}$ edge is the start point of the $(j + 1)^{th}$ edge. The start point of the first edge in a given walk is the start point of the first edge in the walk and, analogously, the end point of the walk is the end point of the last edge in the walk. The \emph{length} of a walk $W$, denoted $\vert W \vert$ is the number of edges in the walk. A graph $G = (V, E)$ is \emph{connected} if there exists, for every pair of vertices $v, u \in V$, a walk starting at $v$ and ending at $u$. Note that a single walk $W$ may contain multiple copies of the same edge without contradiction. The \emph{distance} between two vertices $v, u \in V$ in a graph, denoted $\dist(v,u)$ is the value such there exists some walk $W$ of length $\dist(v, u) $ with the start point $v$, the end point $u$ and $\forall W' \in \{W''$ is a walk in $G$ starting at $v$ and ending at $u \}$, $\vert W' \vert \geq \dist(v, u) $. The \emph{diameter} of a graph $G$ is the maximum distance between any pair of points in the graph, formally the value $\max_{v, u \in V} \dist(v, u)$. 

The \emph{neighbourhood} of a vertex $v$, denoted $N(v)$ is the set of vertices $N(v) = \{u \in V \mid (v, u) \in E \}$. The \emph{degree} of a vertex $v$, denoted $\Delta(v)$ is the number of vertices in the neighbourhood of $v$, formally, $\Delta(v) = \vert N(v) \vert$.

A \emph{directed graph} is a graph where each edge has an \emph{orientation}, with a fixed start point and end point. Thus, the edge $(v_i, v_j)$ is not equivalent to $(v_j, v_i)$. A walk on a directed graph is defined in the same way as a walk on an undirected graph. In this paper, we restrict ourselves to \emph{symmetric directed graphs}. A directed graph $G = (V, E)$ is \emph{symmetric} if the edge $(v_i, v_j) \in E$ iff $(v_j, v_i) \in E$. The \emph{in-edges }of a vertex $v$ is the set of edges in $E$ with the end point $v$, formally, $\inedge(v) = \{(u, v) \in E\}$. Analogously, the \emph{out-edges} of a vertex $v$ is the set of edges in $\out(v) = \{(v, u) \in E\}$, i.e. the set of vertices with the start point $v$.

An \emph{edge-weighted graph} (resp., \emph{edge-weighted directed graph}) is a graph equipped with a \emph{weighting} function $F: E \mapsto \mathbb{N}$, associating some integer cost to each edge in the graph.

A graph of $n$-vertices is a \emph{tree} if the graph contains $n - 1$ edges and is connected. Given a graph $G = (V, E)$, a \emph{spanning tree} of $G$ is a tree $T = (V, E')$ over the same set of vertices as $V$ with an edge set $E' \subseteq E$. The \emph{weight} of a spanning tree $T = (V, E')$ on the graph $G$ weighted by the function $F$ is $\sum_{e \in E'} F(e)$. A \emph{minimum weight spanning tree} of a weighted graph $G$ weighted by the function $F$, is a spanning tree $T = (V, E')$ such that, letting $\mathcal{T} = \{(V, E_1), (V, E_2), \dots, (V, E_{\tau})\}$ be the set of all spanning trees of $G$, $\sum_{e \in E'} F(e) \leq \sum_{e' \in H} F(e'), \forall (V, H) \in \mathcal{T}$.

\paragraph*{Temporal Graphs}

A \emph{temporal graph} is a generalisation of a graph, herein called a \emph{static graph} whenever confusion may otherwise arise, where, rather than having a single edge set, the graph contains an ordered sequence of timesteps, by convention $E_1, E_2, \dots, E_T$, while maintaining a shared set of vertices. We call each edge set a \emph{timestep} (also known as a snapshot). An edge $e$ is \emph{active} at timestep $t$ if $e \in E_t$. We assume, without loss of generality, that each timestep contains at least one active edge. In our definition, an edge $e$ may be active in any number of timesteps.

By default, we assume that the temporal graph is undirected, with a temporal graph referred to as a \emph{directed temporal graph} if the edge sets are directed. We assume that either every edge set is directed, or none are. Given a temporal graph $\mathcal{G}$ the \emph{lifetime} of the graph is the number of timesteps, thus, if $\mathcal{G} = (V, E_1, E_2, \dots, E_T)$, then the lifetime of the graph is $T$.

The \emph{underlying graph} of a temporal graph $\mathcal{G} = (V, E_1, E_2, \dots, E_T)$, denoted $U(\mathcal{G})$, is the static graph $U(\mathcal{G}) = (V, \bigcup_{t \in [1, T]} E_t)$, i.e. the static graph formed with the edge set corresponding to the union of all timesteps in the graph. A temporal graph is a \emph{symmetric directed temporal graph} if the underlying graph is a symmetric directed graph. In general, given some property $X$ of a static graph, a temporal graph has the property $X$ if the underlying graph satisfies the property. A temporal graph \emph{always satisfies} a given property $X$ if the static graphs $G_1 = (V, E_1)$, $G_2 = (V, E_2)$, $\dots$, $G_T = (V, E_T)$ also satisfy this property. Notably, a static graph $G = (V, E)$ is \emph{connected} if, for each pair of vertices $v_i, v_j \in V$, there exists some walk starting at $v_i$ and ending at $v_j$. Therefore, a temporal graph $\mathcal{G} = (V, E_1, E_2, \dots, E_{T})$ is an \emph{always connected temporal graph} if the static graph $G_t = (V, E_t)$ is connected, for every $t \in [1, T]$, where $T$ is the lifetime of $\mathcal{G}$.

The \emph{frequency} of an edge $e$, denoted $f_e$, is the minimum number of timesteps such that $e$ is active at least once every $f_e$ timesteps. Formally, $f_e$ satisfies $\forall t \in [1, T + 1 - f_e]$, $e \in \bigcup_{t' \in [t, t + f_e - 1]} E_{t'}$ and, $\forall f' \in [1, f_e - 1]$, $\exists t \in [1, T + 1 - f']$ such that $e \notin \bigcup_{t' \in [t, t + f' - 1]} E_{t'}$. A temporal graph is \emph{ $f$-frequent} if $f_e \leq f$, $\forall e \in E$.

An edge $e$ is \emph{$r$-regular} if given any $t \in [r + 1, T - r]$, either $e$ is active in all of $E_{t - r}$, $E_t$, and $E_{t + r}$, or $e$ is inactive in all of $E_{t - r}$, $E_t$, and $E_{t + r}$. The \emph{regularity} of an edge $e$, denoted $r_e$, is the value such that $e$ is $r_e$ regular and, $\forall r' \in [1, r_e - 1]$, $e$ is not $r'$ regular. A temporal graph is \emph{$r$-regular} if $r_e \leq r$, $\forall e \in E$.  Observe that $r_e \geq f_e$. Therefore, every $x$-regular temporal graph is an $x$-frequent temporal graph.

A \emph{temporal walk} is an set of edge-timestep tuples, $\mathcal{W} = ((v_{i_1}, v_{i_2}), t_1)$, $((v_{i_2}, v_{i_3}), t_2)$, $\dots$, $((v_{i_{m - 1}}, v_{i_m}), t_{m - 1})$ such that:
\begin{itemize}
    \item $(v_{i_1, i_2}) (v_{i_2}, v_{i_3}) \dots (v_{i_{m - 1}, v_{i_m}})$ form a walk in the underlying graph $U(\mathcal{G})$,
    \item the edge $(v_{i_j}, v_{i_{j + 1}})$ is active in timestep $t_j$, and,
    \item $t_1 < t_2 < \dots < t_{m - 1}$.
\end{itemize}
The \emph{length} of a temporal walk $\mathcal{W}$, denoted $\vert \mathcal{W} \vert$ is the timestep $t_{m - 1}$ associated with the final edge. We refer to the $i^{th}$ tuple of a temporal walk as the $i^{th}$ \emph{step} in the walk.

\paragraph*{Exploration}

An \emph{exploration} of a static graph is a walk $W$ such that there exists, for every vertex $v \in V$ at least one edge in $W$ containing $v$. Analogously, an exploration of a temporal graph is a temporal walk $\mathcal{W}$ such that there exists, every vertex $v \in V$ , at least one edge in $\mathcal{W}$ containing $v$.

\begin{problem}[Temporal Graph Exploration Problem]
    \label{prob:exploration_problem_decide}
    Given a temporal graph, $\mathcal{G} = (V, E_1, E_2, \dots, E_T)$, a vertex $v \in V$ and length $\ell \in [1, T]$, does there exist a temporal walk $\mathcal{W}$ starting at $v$ with a length less than $\ell$?
\end{problem}

We can rewrite this as an optimisation problem as:

\begin{problem}[Fastest Temporal Graph Exploration Problem]
    \label{prob:exploration_problem_optimise}
    Given a temporal graph, $\mathcal{G} = (V, E_1, E_2, \dots, E_T)$, a vertex $v \in V$ what is the temporal walk $\mathcal{W}$ starting at $v$ and exploring $\mathcal{G}$ such that, for any alternative walk $\mathcal{W}'$ either $\mathcal{W}'$ does not explore $\mathcal{G}$ or $\vert \mathcal{W} \vert \leq \vert \mathcal{W}' \vert$?
\end{problem}

As noted in the introduction, both problems are known to be NP-hard. Therefore, this paper focuses on providing upper bounds for several classes of graphs such that every graph has a walk exploring it in $f(\mathcal{G})$ timesteps for some stated function $f$.



\section{Frequent and Regular Edges}

\label{sec:main_results}

In this section, we provide our primary algorithmic results, covering the our algorithms for exploring $f$-frequent temporal graphs and $r$-regular temporal graphs. At a high level, our approach is as follows. We build a minium weight spanning tree on the edge-weighted static graph corresponding to the underlying graph, with the weight on each edge corresponding to the frequency of the edge in the temporal graph. First, we provide technical results on determining the frequency of each edge.

\begin{lemma}
    \label{lem:getting_frequency}
    Given a temporal graph $\mathcal{G} = (V, E_1, E_2, \dots, E_T)$, with the underlying graph $U(\mathcal{G}) = (V, E)$, the frequency $f_e$ of the edge $e \in E$ can be determined in $O(T)$ time.
\end{lemma}

\begin{proof}
    Note that $f_e$ corresponds to the number of timesteps in the longest consecutive number of timesteps that do not contain $e$, plus one.
    To determine this value, we take a brute force approach, checking each timestep in order, while maintaining a pair of variables corresponding to the longest gap without a timestep containing $e$, and the number of timesteps since the last timestep containing $e$. Formally, let $\ell$ be the largest number of consectutive timesteps not containing $e$ found so far, and let $C$ be the number of timesteps since the last time $e$ was active, with both initialy set to $0$. For each $t \in [1, T]$, in order from $1$ to $T$, we do one of the following:
    \begin{itemize}
    	\item if $e \notin E_t$, we increment the value of $C$, setting it to $C + 1$, then update the value of $\ell$ to $\max(\ell, C)$.
    	\item otherwise, if $e \in E_t$, we reset $C$ to $0$.
    \end{itemize}
    Once this proccess has completed, note the $\ell$ must contain the length of the longest consecutive sequence of timesteps not containing $e$, formally, the value such that $\exists t \in [1, T - \ell + 1]$ s.t. $e \notin \bigcup_{t' \in [t, t + \ell - 1]} E_{t'}$ and, for any $t \in [1, T - \ell]$, $e \in \bigcup_{t' \in [t, t + \ell]} E_{t'} $. Therefore, the frequency of $e$ is $\ell + 1$, giving the algorithm.
    See Algorithm \ref{alg:calculating_frequency} for pseudocode corresponding to this algorithm.
\end{proof}

\begin{algorithm}[ht]
    \caption{Algorithm for computing the frequency of the edge $e$ in the temporal graph $\mathcal{G} = (V, E_1, E_2, \dots, E_T)$ following Lemma \ref{lem:getting_frequency}.}
    \label{alg:calculating_frequency}
    \begin{algorithmic}
        \Procedure{ComputeEdgeFrequency}{Temporal Graph $(V, E_1, E_2, \dots, E_T)$, edge $e \in \bigcup_{t \in [1, T]} E_t$}
            \State $\ell \gets 0$
            \State $C \gets 0$
            \For{$t \in [1, T]$} \Comment{\emph{We assume $t$ iterates over the set $[1, T]$ in order.}}
                \If{$e \in E_t$}
                    \State $C \gets 0$
                \Else
                    \State $C \gets C + 1$
                    \State $\ell \gets \max(\ell, C)$
                \EndIf
            \EndFor
            \State \Comment{\emph{Note that $\ell$ is the longest run of timesteps not containing $e$. Therefore, the frequency of $e$ is $\ell + 1$, as $e$ must appear at least once every $\ell + 1$ timesteps.}}
            \State \textbf{return} $\ell + 1$
        \EndProcedure
    \end{algorithmic}
\end{algorithm}

\begin{corollary}
    \label{col:finding_the_total_frequency}
    Given a temporal graph $\mathcal{G} = (V, E_1, E_2, \dots, E_T)$, let $F$ be an array over the edge set $E = \bigcup_{t \in T} E_t$ such that given some $e \in E$, $F[e]$ is the frequency of the edge $e$ in $\mathcal{G}$. Then, the array $F$ may be computed in $O(\vert E \vert \cdot T)$ time.
\end{corollary}

\begin{figure}
    \centering
    \begin{tabular}{c:c:c}
         \includegraphics[width=0.25\linewidth]{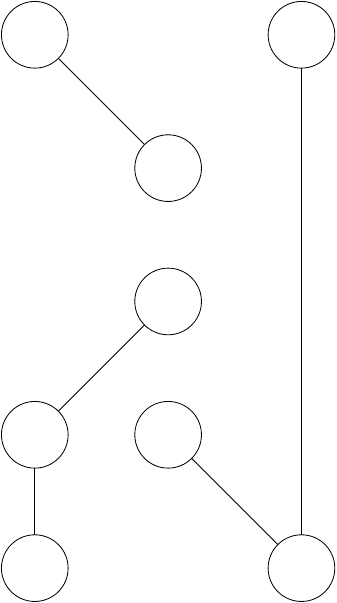} & \includegraphics[width=0.25\linewidth]{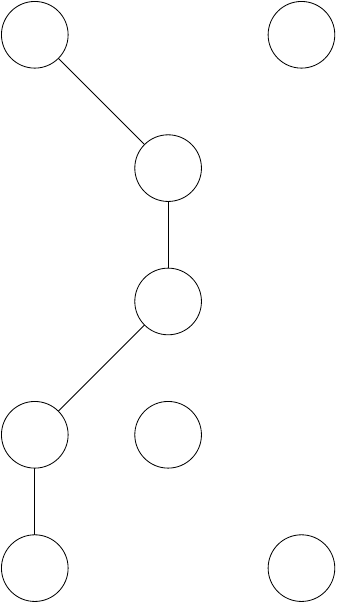} & \includegraphics[width=0.25\linewidth]{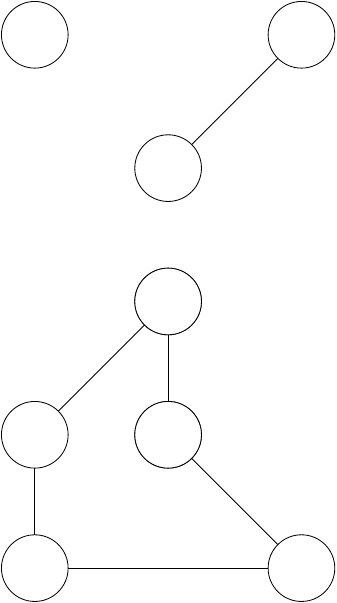} \\
         \emph{Timestep 1 ($E_1$)} & \emph{Timestep 2 ($E_2$)} & \emph{Timestep 3 ($E_3$)}\\
         \hdashline
         \includegraphics[width=0.25\linewidth]{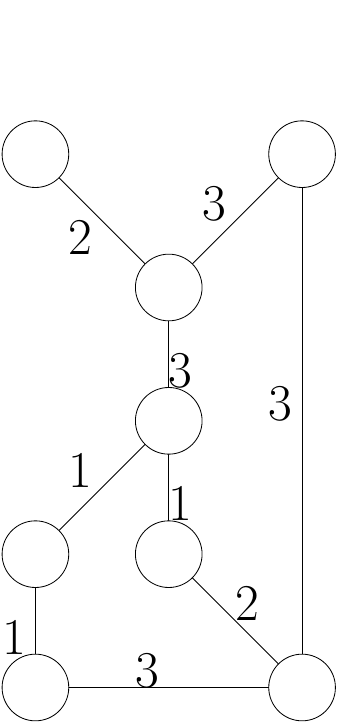} & \includegraphics[width=0.25\linewidth]{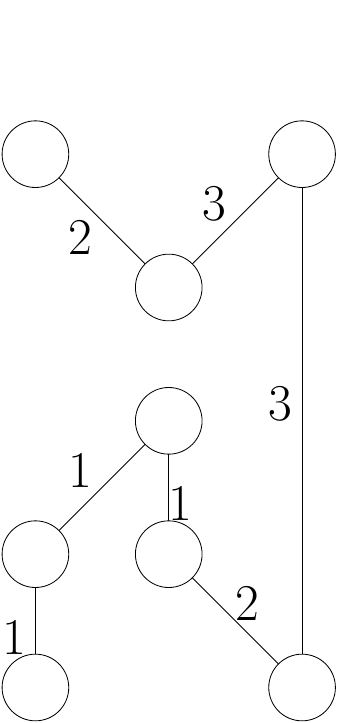} & \includegraphics[width=0.25\linewidth]{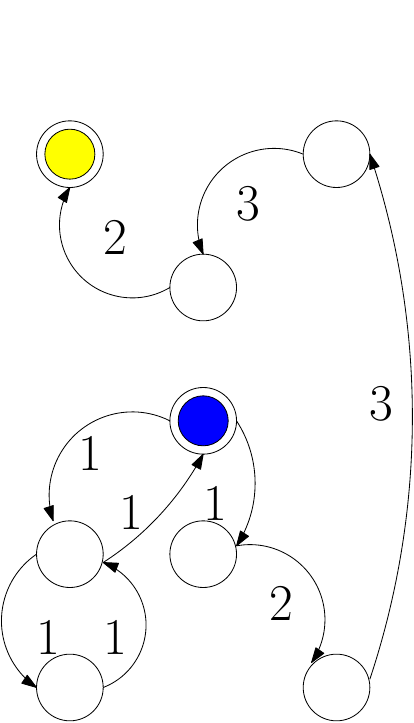}\\
         \emph{$FW(\mathcal{G})$} & \emph{Minium Weight Spanning Tree on $FW(\mathcal{G})$} & \emph{Exploration of $\mathcal{G}$}
    \end{tabular}
    \caption{Overview of the exploration of the temporal graph $\mathcal{G}$ with $\mathcal{G} = (V, E_1, E_2, E_3, E_1, E_2, E_3, \dots, E_1, E_2, E_3)$, i.e. the temporal graph formed by repeating the edge sets $E_1$, $E_2$ and $E_3$ some number of times.}
    \label{fig:Exploration_Overview}
\end{figure}




We now show how to explore temporal graphs using weighted spanning trees as a basis. Informally, we construct a minimum weight spanning tree over the static graph formed by weighting each edge on the underlying graph by the frequency of the edge in the temporal graph. This tree is used to devise an exploration schedule by finding a walk exploring this tree, with the total weight of this walk corresponding to the maximum length of the exploration of the temporal graph.

\begin{lemma}[\cite{kruskal1956shortest}]
    \label{lem:computing_spanning_tree}
    A minimum weight spanning tree on the weighted static graph $G = (V, E)$ can be constructed in $O(\vert E \vert \log \vert V \vert)$ time.
\end{lemma}

Given a temporal graph $\mathcal{G} = (V, E_1, E_2, \dots, E_t)$, let $FW(\mathcal{G}) = (V, E, W)$ (read \textbf{F}requency \textbf{W}eighted) be the edge-weighted static graph such that $V, E = U(\mathcal{G})$ and $W : E \mapsto [1, T]$ is the weighting function over $E$ with $W(e) = f_e$, where $f_e$ is the frequency of the edge $e \in E$.

\begin{lemma}[Folklore]
    \label{lem:explorling_trees}
    Given a tree $T$ with $n$ vertices, there exists a walk exploring $T$ of length at most $2n - 3$ starting at any vertex $v$. Further, this walk can be found in $O(n)$ time.
\end{lemma}

\begin{theorem}
    \label{thm:frequent_exploration}
    Given some temporal graph $\mathcal{G} = (V, E_1, E_2, \dots, E_T)$ with the underlying graph $U(\mathcal{G}) = (V, E)$, $\mathcal{G}$ can be explored, starting from any vertex $v \in V$, in $2 F$ timesteps, where $F$ is the weight of the minimum weight spanning tree on $FW(\mathcal{G})$. Further, this walk can be determined in $$O\left(T \left\vert E \right\vert \right)$$ time.
\end{theorem}

\begin{proof}
    First, we construct the minimum weight spanning tree $T$ on the static graph $FW(\mathcal{G})$, and a walk, $W = e_1, e_2, \dots, e_m$, exploring $T$ and starting at $v$. We convert $W$ into a temporal walk $\mathcal{W}$ as follows.
    Let $t_1$ be the index of first timestep in which $e_1$ is active. Formally, $e_1 \in E_{t_1}$ and, $\forall t' \in [1, t_1 - 1]$, $e_1 \notin E_{t'}$.
    We set the tuple $(e_1, t_1)$ as the first step in the walk $\mathcal{W}$.

    We determine the remaining steps on $\mathcal{W}$ in an iterative manner, using the value of the $(i - 1)^{th}$ step to determine the $i^{th}$ step. Formally, if the $(i - 1)^{th}$ tuple of $\mathcal{W}$ is $(e_{i - 1}, t_{i - 1})$, then the $i^{th}$ step is $(e_i, t_i)$ where $t_i \in [t_{i - 1} + 1, T]$ is the index such that $e_i \in E_{t_i}$ and $\forall t' \in [t_{i - 1} + 1, t_i - 1]$, $e_i \notin E_{t'}$.

    To determine the total length of this path, observe that, by definition, if the edge $e_i$ is not active at timestep $t_{i - 1} + 1$, the there must exist some $t_i \in [t_{i - 1} + 2, t_{i - 1} + 1 + f_{e_i}]$ such that $e_i$ is active. Therefore, $t_m \leq \sum_{e \in \{e_1, e_2, \dots, e_m\}} f_e$. Further, as the weight of the spanning tree $F = \sum_{e \in \{e_1, e_2, \dots, e_m\}} f_e$, we get $t_m \leq 2F$.
    
    To determine the time complexity, note first that $U(\mathcal{G})$ can be computed in $O(\sum_{t \in [1, T]}\vert E_t \vert)$ time. Next, we can compute, by Corollary \ref{col:finding_the_total_frequency}, the frequency of each edge in $O(\vert E \vert T)$ time. Further, we have, by Lemma \ref{lem:computing_spanning_tree}, that the spanning tree can be computed in $O(\vert E \vert \log(\vert V \vert))$ time. From Lemma \ref{lem:explorling_trees}, the walk $W$ can be computed in $O(\vert V \vert)$ time.
    Finally, noting that we can determine if $e_i \in E_t$ in $O(1)$ time, for any $e_i \in E, t \in [1, T] $, and at most one such check is made per timestep, the temporal walk $\mathcal{W}$ can be computed from $W$ in $O(T)$ time. Therefore, the total complexity is $O(\sum_{t \in [1, T]}\vert E_t \vert + \vert E \vert T + \vert E \vert \log \vert V \vert + T)$ time. As $\vert E \vert T \geq  \sum_{t \in [1, T]}\vert E_t \vert \geq T$ and, in order to explore all $\vert V \vert$ vertices $T \geq \vert V \vert$, we can simplify this to $O\left(T \left\vert E \right\vert \right)$ time, giving the statement.
\end{proof}

We provide an overview of the process outlined in Theorem \ref{thm:frequent_exploration} in Figure \ref{fig:Exploration_Overview}, and pseudocode in Algorithm \ref{alg:exploring_frequent}. We can simplify this for $f$-frequent temporal graphs to get the following simplified bounds.

\begin{algorithm}
    \caption{Algorithm for computing a temporal walk exploring a temporal graph $\mathcal{G} = (V, E_1, E_2, \dots, E_T)$ using the frequency of the edges as per Theorem \ref{thm:frequent_exploration}.}
    \label{alg:exploring_frequent}
    \begin{algorithmic}
        \Procedure{ExploreTemporal}{Temporal Graph $\mathcal{G} = (V, E_1, E_2, \dots, E_T)$, Start Vertex $v$}
            \State $G = FW(\mathcal{G})$ \Comment{Make the frequency weighted static graph $G$ of $\mathcal{G}$.}
            \State $T = \MWST(G)$ \Comment{Find the minimum weight spanning tree $T$ of $G$.}
            \State $e_1, e_2, \dots, e_M \gets$ \textsc{Explore}($T, v$) \Comment{Find a walk $W$ exploring $T$ starting at $v$ with at most $2n - 3$ edges.}
            \State $\mathcal{W} \gets \emptyset$ \Comment{Set $\mathcal{W}$ as the temporal walk exploring $\mathcal{G}$, initially empty.}
            \State $t \gets 1$ \Comment{Variable to store the timestep of the next step in $\mathcal{W}$, starting with the first.}
            \For{$e \in e_1, e_2, \dots, e_m$} \Comment{Iterate over the edges in the walk $W$ in order.}
                \While{$e \notin E_t$} \Comment{Increment the value of $t$ until we reach a timestep with $e$ active.}
                    \State $t \gets t + 1$
                \EndWhile
                \State Append $(e, t))$ to $\mathcal{W}$ \Comment{Append the tuple $(e, t)$ to the end of the walk $\mathcal{W}$.}
                \State $t \gets t + 1$ \Comment{Increment $t$ one last time to avoid adding multiple edges to $\mathcal{W}$ one the same timestep.}
            \EndFor
            \State \textbf{return} $\mathcal{W}$
        \EndProcedure
    \end{algorithmic}
\end{algorithm}

\begin{lemma}
    \label{lem:weakly_f_frequent}
    Any $f$-frequent temporal graph $\mathcal{G} = (V, E)$ with $n$ vertices can be explored, starting at any vertex $v \in V$, in $f(2n - 3)$ timesteps and, further, a temporal walk exploring $\mathcal{G}$ in $O(\vert E \vert (\log \vert V \vert + T))$ time.
\end{lemma}

\begin{proof}
    Utilising the same approach as Theorem \ref{thm:frequent_exploration}, we construct the temporal walk $\mathcal{W}$ exploring $\mathcal{G}$. Note that each edge has a frequency of at most $f$, thus the weight of the minimum weight spanning tree is at most $nf$. Observing that, by Lemma \ref{lem:explorling_trees}, we can explore this tree in $O(2n - 3)$ steps, the length of the temporal walk $\mathcal{W}$ is at most $f(2n - 3)$, corresponding to waiting $f - 1$ steps at each of the $2n - 3$ edges traversed in the walk, followed by one timestep traversing the edge.
\end{proof}

\begin{corollary}
    \label{col:weakly_f_symetric_frequent}
    Any $f$-frequent symmetric directed temporal graph $\mathcal{G} = (V, E)$ with $n$ vertices can be explored, starting at any vertex $v \in V$, in $f(2n - 3)$ timesteps and, further, a temporal walk exploring $\mathcal{G}$ in $O(\vert E \vert (\log \vert V \vert + T))$ time.
\end{corollary}

\begin{proof}
    Note that the same approach as in Lemma \ref{lem:weakly_f_frequent} can be applied directly to $f$-frequent symmetric directed temporal graphs as, rather than traversing the same $f$-frequent edge twice, we traverse two $f$-frequent edges once each.
\end{proof}

Recalling that an edge with regularity $r_e$ is, by definition, $r_e$ frequent, we get the following.

\begin{corollary}
    \label{col:regular_edges}
    Any $r$-regular temporal graph with $n$ vertices can be explored in $r(2n - 3)$ timesteps from any starting vertex.
\end{corollary}

\begin{corollary}
    Any $r$-regular symmetric directed temporal graph with $n$ vertices can be explored in $r (2 n - 3)$ timesteps from any starting vertex.
\end{corollary}

We can contrast this against the natural lower bound, showing that there exists an $r$-regular graph requiring $r(2n - 5) + 1$ timesteps to explore.

\begin{theorem}
    For any regularity $r \in \mathbb{N}$ and number of vertices $n \in \mathbb{N}$, there exists some $r$-regular temporal graph requiring $r(2n - 5) + 1$ timesteps to explore.
\end{theorem}

\begin{proof}
    Consider the star graph $S_n = (V, E)$, defined as the graph containing the vertex set $v_1, v_2, \dots, v_n$ and the edge set $E = \{(v_1, v_i) \mid i \in [2, n]\}$. Let $\mathcal{S}_n = (V, E_1, E_2, \dots, E_{f(2n - 1)})$ be the temporal graph such that at timestep $E_t$ either $E_{t} = E$, if $t \bmod r\equiv 1$, or $\{v_2, v_1\}$ otherwise. Further, let us assume that the agent starts at the vertex $v_2$.
    Observe that to explore this graph, the agent must traverse all but two edges in the graph twice, those being the edge $(v_2, v_1)$ and some other edge $(v_1, v_i)$. Note that we may assume, without loss of generality, that the agent traverses the edge $(v_2, v_1)$ in the first timestep. Once the agent traverses an edge, it must wait $r - 1$ timesteps for this edge to reactivate. Thus, once the agent has reached the vertex $v_1$ after traversing the edge $(v_{i - 1}, v_1)$, for some $v_{i - 1} \in V \setminus\{v_1\}$, reaching the vertex $v_i$ from $v_1$ and returning requires $2r$ timesteps, corresponding to $r - 1$ timesteps waiting for $(v_1, v_i)$ to become active, one timestep to traverse $(v_1, v_i)$, $r - 1$ timesteps for $(v_i, v_1)$ to become active, and one timestep to traverse the edge. As there are $n - 2$ such vertices, and recalling that we do not need to return to $v_1$ after reaching the last vertex in the exploration, we get a total of $2r(n - 3) + r + 1 = r(2n - 5) + 1$ timesteps needed to explore this temporal graph.
\end{proof}

\begin{corollary}
    For any frequency $f \in \mathbb{N}$ and number of vertices $n \in \mathbb{N}$, there exists some $f$-frequent temporal graph requiring $f(2n - 5) + 1$ timesteps to explore.
\end{corollary}

\section{Restricted classes with Frequent Edges}
\label{sec:motivating_results}

In this section, we provide a set of results on some subsets of temporal graphs which may be reduced to temporal graphs with frequent edges.

\subsection{Public Transport Graphs}

We first consider \emph{public transport graphs}. We use such temporal graphs to model public transport timetables. Here, the temporal graph is composed of a set of temporal walks, each representing some transport line with vertices representing stops. For example, one may represent the bus stops and lines in a city by the vertices and walks respectively. We assume that once a walk is completed, it will be repeated immediately, noting that we can simulate a delay by adding a self-loop at the end of the walk, active at an appropriate timestep. 

Formally, we define a public transport graph by the set of temporal walks $\mathcal{W}_1$, $\mathcal{W}_2$, $\dots$, $\mathcal{W}_m$ over a common set of vertices $V$ and given lifespan $T$. Letting $L_i = \vert \mathcal{W}_i\vert$ be the last timestep containing some edge in the $i^{th}$ walk, and $\mathcal{W}_i[j]$ denote the $j^{th}$ edge in walk $\mathcal{W}_i[j]$, we have, at timestep $t \in [1, T]$, the edge $e \in E_t$ iff $\exists i \in [1, m], j \in [1, \vert \mathcal{W}_i \vert]$ s.t. $e = \mathcal{W}_i[j]$ and $\mathcal{T}_i[j] = t \bmod L_i$. Informally, $e$ is active in timestep $t$ if there exists some $\mathcal{W}_i$ such that $e$ is active in timestep $t$ (modulo the length of the walk). Note that an edge may be used in multiple routes, possibly even in the same timestep.

\begin{theorem}
    Given a public transport graph $\mathcal{G} = (V, E_1, E_2, \dots, E_T)$ composed of the routes $\mathcal{W}_1, \mathcal{W_2}, \dots, \mathcal{W}_m$, such that $U(\mathcal{G})$ is a single connected component, there exists an exploration of $\mathcal{G}$ requiring $(2 n - 3)\max_{i \in [1, m]} L_i$ timesteps, where $L_i = \vert W_i \vert$.
\end{theorem}

\begin{proof}
    Observe that, by construction, each edge must have frequency at least $L_i$. Therefore, using Theorem \ref{thm:frequent_exploration} and Lemma \ref{lem:weakly_f_frequent}, we get the bound.
\end{proof}

\subsection{Sequential Connection Graphs}

We now consider the class of graphs where each vertex activates the set of incident edges in a fixed, repeating order. One can think of these as graphs in which each vertex allows connections to one neighbour in turn. For example, a radio receiver switching between a set of bands in a regular order.

Formally, a symmetric directed temporal graph $\mathcal{G} = (V, E_1, E_2, \dots, E_T)$ is a \emph{sequential connection graph} if there exists, for every vertex $v \in V$ there exists a permutation of the incoming edges to $v$, $P_v = e_1, e_2, \dots, e_{\Delta(v)}$, such that at timestep $t$, $P_v[t \bmod \Delta(v)] \in E_t$ and, $\forall e \in \inedge(v) \setminus\{P_v[t \bmod \Delta(v)]\}$, $e \notin E_t$.

\begin{theorem}
    Given a sequential connection graph $\mathcal{G} = (V, E_1, E_2, \dots, E_T)$, $\mathcal{G}$ can be explored in $2 \sum_{v \in V} \Delta(v) = 4 \vert E \vert$ timesteps, where $U(\mathcal{G}) = (V, E)$.
\end{theorem}

\begin{proof}
    Note that the frequency of the edge $(v_i, v_j)$ is $\Delta(v_j)$. Therefore, by Theorem \ref{thm:frequent_exploration}, we get the bound.
\end{proof}

\subsection{Broadcast Networks}

In distributed computing, the \emph{broadcast} model of communication is one in which, at each communication round, every vertex either sends the same message to every neighbour, or sends no message \cite{korhonen2018deterministic}. We can use a temporal graph to model communication in these networks by representing each communication round as a timestep in the temporal graph, with messages from vertex $v$ to vertex $u$ in timestep $t$ represented by having the edge $(v, u)$ in $E_t$. We call these graphs \emph{Broadcast Networks}.

We have two additional restrictions on these temporal graphs. First, to model the restrictions on the broadcast model, we require that each time a vertex broadcasts a message, it does so to all neighbours. 
Formally, $\forall v \in V, t \in [1, T]$, either $\{(v, u) \mid u \in N(v) \} \subseteq E_t$ or $\{(v, u) \mid u \in N(v) \} \cap E_t = \emptyset$. Second, we assume that a vertex will only send the next message once it has received some sort of acknowledgement from every neighbour, either in the form of a simple acknowledged message, or a full message from the neighbour containing the round in which the neighbour has sent the message. This avoids a desynchronisation of the network, where one vertex progresses despite no follow-up information being given. To simplify the model, we assume that if two neighbours send a message on the same round, both are consider in sync, and thus do not have to wait for the other to send the next message. Formally, we say that an edge $(v, u)$ may be active in timesteps $t_1$ and $t_2$, where $t_1 < t_2$ iff there exists some $t' \in [t_1, t_2 - 1]$ where $(u, v) \in E_{t'}$.

For convenience, we say a vertex $v$ is active in timestep $E_t$ iff $\{(v, u) \mid u \in N(v) \} \subseteq E_t$, and use the notation $A_t$ to denote the set of vertices active in timestep $t$. Thus, the second condition can be rewritten as $v \in A_{t_1} \cap A_{t_2}$ where $t_1 < t_2$ iff $\forall u \in N(v)$, $\exists t_u \in [t_1, t_2 - 1]$ s.t. $u \in A_{t_u}$.


We combine these restrictions to derive the frequencies of the edges.

\begin{lemma}
    Given a broadcast network $\mathcal{G} = (V, E_1, E_2, \dots, E_T)$, the edge $(u, v)$ has a frequency of $d n$, where $n$ is the number of vertices in $V$ and $d$ is the diameter of $U(\mathcal{G})$.
\end{lemma}

\begin{proof}
    Let $v$ be some vertex in the graph $\mathcal{G}$ that is active at timesteps $t$ and $t'$, where $v \notin A_i$, $\forall i \in [t + 1, t' - 1]$. Observe that each vertex $u \in N(v)$ can be active at most once in $A_{t + 1}, A_{t + 2}, \dots, A_{t' - 1}$. Further, given some vertex $u' \in N(u) \setminus \{ v \}$, $u'$ can be active at most twice in $A_{t + 1}, A_{t + 2}, \dots, A_{t' - 1}$, once before $u$ is active and once after.

    Now, consider some vertex $v'$ at a distance $\ell$ from $v$, and let $u_1, u_2, \dots, u_{\ell - 1}$ satisfy $(v, u_1) \in E$, $(u_{\ell - 1}, v') \in E$ and $(u_i, u_{i + 1}) \in E$, $\forall i \in [1, \ell - 2]$, where $U(\mathcal{G}) = (V, E)$. Further, let $\tau_{x}$ be the number of timesteps between $t + 1$ and $t' - 1$ in which the vertex $x$ is active. Note that $\tau_{u_1} = 1$, and $\tau_{u_2} \leq 2$. In general $\tau_{u_i} \leq \tau_{u_{i - 1}} + 1$ as the vertex $u_{i}$ can be active once before $u_{i - 1}$, once between each activation of $u_{i - 1}$, and once after. Therefore, $\tau_{v'} \leq \dist(v, v')$.

    As $\dist(v, v') \leq d$, $\forall v' \in V$, no vertex may be active more than $d$ times between $t$ and $t'$. Hence, as there are $n$ vertices in $V$, there can be at most $dn$ vertices activated between timesteps $t$ and $t'$, thus $t' - t \leq d n$. By extension, the frequency of $v$ is at most $d n$.
\end{proof}

\begin{corollary}
    \label{col:broadcast_exploration}
    Any broadcast network $\mathcal{G} = (V, E_1, E_2, \dots, E_T)$ with a diameter $d$ can be explored in $d n(2n - 3)$ timesteps.
\end{corollary}

\paragraph*{Always Connected Broadcast Networks}

We now consider the restriction to broadcast networks where the temporal graph is connected at each timestep. We show that these temporal graphs can be explored in $(\delta + 1)(2 n - 3)$ timesteps where $\delta = \min_{v \in V} \Delta(v)$, i.e. $\delta$ is the lowest degree of any vertex in the underlying graph, with $\Delta(v)$ denoting the degree of $v$ in the underlying graph.

For simplicity, in this section we use the term \emph{frequency} with reference to vertices, denoting the frequency of a vertex $v$ by $f_v$ and defining $f_v$ as the smallest value such that $\forall t \in [1, T + 1 - f_v]$, $v \in \bigcup_{t' \in [t, t + f_v - 1]} A_{t'}$, i.e. the smallest value such that $v$ is active at least once every $f_v$ timesteps. Note that this is analogous to the definition of edges and, indeed, in this model, the frequency of the edge $(v, u)$ is equal to the frequency of the vertex $v$.

\begin{lemma}
    \label{lem:degree_and_frequency_broadcast}
    Given an always connected broadcast network $\mathcal{G} = (V, E_1, E_2, \dots, E_t)$ with underlying graph $U(\mathcal{G}) = (V, E)$, the edge $e \in E$ has frequency $\delta + 1$, where $\delta = \min_{v \in V} \Delta(v)$.
\end{lemma}

\begin{proof}
    Consider some vertex $v \in V$ that is active in timestep $t$. In order for $v$ to be connected, there must, at each time step, be some vertex $u \in N(v)$ that is active, or $v$ itself must be active. As each vertex in $N(v)$ can be active at most once before $v$ is active, $v$ must be active in at least one timestep in $[t + 1, t + 1 + \Delta(v)]$. Thus $v$ has a frequency of $\Delta(v) + 1$.
    Now, consider some pair $(v, u)$ such that $(v , u) \in E$. As $v$ has a frequency of at most $\Delta(v) + 1$ and $u$ has a frequency of at most $\Delta(u) + 1$, to satisfy the condition that $v$ is active at most once between activations of $u$ (and $u$ is active at most once between the activations of $v$), $v$ must also be active at least once every $\Delta(u) + 1$ timesteps. Hence $f_v \leq \min(\Delta(v), \Delta(u)) + 1$. Extrapolating this, we have that $f_v \leq \min_{u \in N(v)} \Delta(u) + 1$ and $f_{v} \leq \Delta(v) + 1$. Further, we have by the same arguments, that $f_v \leq \min_{u \in N(v)} f_u$. Expanding this across the temproal graph gives that $f_v \leq \min_{u \in V \setminus\{v \}} f_u \leq \min_{u \in V \setminus\{v\}} \Delta(u) + 1$. Hence, $f_v \leq \min_{u \in V} \Delta(u) + 1 = \delta + 1$.
\end{proof}

\begin{theorem}
    \label{thm:always_connected_broadcasts}
    Given an always connected broadcast network $\mathcal{G} = (V, E_1, E_2, \dots, E_t)$, $\mathcal{G}$ can be explored in $(\delta + 1)(2n - 3)$ timesteps, where $\delta = \min_{v \in V} \Delta(v)$.
\end{theorem}

\begin{proof}
    Follows from Theorem \ref{thm:frequent_exploration}, Lemma \ref{lem:weakly_f_frequent}, and Lemma \ref{lem:degree_and_frequency_broadcast}
\end{proof}

\section{Conclusion}

In this paper, we have studied the set of temporal graphs defined by having \emph{frequent edges}, both in the general case and in several motivating settings, including public transport networks and the broadcast communication model. The primary result of this paper is that any such graph with $n$ vertices, where each edge is active at least once every $F$ timesteps, can be explored in $F(2n - 3)$ timesteps.

There are several obvious open questions following this paper. First is if the lower and upper bounds for exploration of temporal graphs with frequent edges can be closed, either through improving the exploration algorithm, or by constructing a stronger worst case example.

Another interesting direction is the question of whether the current results for broadcast networks are optimal. We conjecture that there exists an upper bound of the form $O(n d^k)$ for some constant $k$. Alternatively, lower bounds for this setting would be of great interest. Further, it is of interest as to if these techniques can be applied to other settings.

\bibliographystyle{fundam}
\bibliography{bib}

\end{document}